\pdfoutput=1  %arXiv addition for pdflatex, permits no file extensions.
\documentclass[envcountsame]{llncs}
\usepackage{url}
\usepackage{epsf}
\usepackage{graphicx}
\usepackage{amssymb}
\usepackage{amsmath}
\usepackage{fullpage}
\usepackage{color}
\usepackage{hyperref}

\newcommand{\arr}[1]{\ensuremath{\protect\overrightarrow{#1}}}

\begin{document}

\title{Improved Bounds on the Stretch Factor of $Y_4$ \thanks{Supported by NSF grant CCF-0728909.}}
\author{Mirela Damian\and Naresh Nelavalli}
\authorrunning{M. Damian and N. Nelavalli}
\institute{Department of Computing Sciences, Villanova University, Villanova, USA \\ 
\email{mirela.damian@villanova.edu} and \email{nnelaval@villanova.edu} 
%\and
%Department of Computing Sciences, Villanova University, Villanova, USA. \email{nnelaval@villanova.edu}.
 }

%\date{}

\maketitle

\begin{abstract}
We establish an upper bound of $13+8\sqrt{2} \lesssim 4.931$ on the stretch factor of the 
Yao graph $Y^\infty_4$ defined in the $L_\infty$-metric, improving upon the best previously known upper bound 
of $6.31$. 
We also establish an upper bound of $(11+7\sqrt{2})\sqrt{4+2\sqrt{2}} \lesssim 54.62$ on the stretch factor of the 
Yao graph $Y_4$ defined in the $L_2$-metric, improving upon the best previously known upper bound of 
$662.16$. 
\end{abstract}

\section{Introduction}
Let $V$ be a finite set of points in the plane. 
The \emph{directed Yao graph}~\cite{Yao82} with integer parameter $k > 0$,
denoted $\arr{Y_k}$, is defined as follows. At each point $u \in V$, any $k$
equally-separated rays originating at $u$ define $k$ cones.
In each cone, pick a shortest edge $(u,v)$, if there is one, and add to $\arr{Y_k}$
the directed edge $\overrightarrow{(u,v)}$. Ties are broken arbitrarily.
The \emph{undirected Yao graph} $Y_k$ includes all edges of $\arr{Y_k}$ but ignores their directions.
Most of the time we ignore the direction of an edge $(u,v)$. We refer
to the directed version $\overrightarrow{(u,v)}$ of $(u,v)$ only when its
origin ($u$) is important and unclear from the context.
We will distinguish between $Y_k$, the Yao graph in the Euclidean $L_2$
metric, and $Y^\infty_k$, the Yao graph in the $L_\infty$ metric. Unlike
$Y_k$ however, in constructing $Y^\infty_k$ ties are broken by always
selecting the most counterclockwise edge. This tie breaking rule 
was first mentioned in~\cite{BDD+12}, 
where it was required in order to maintain the planarity of $Y_4^\infty$. 
Throughout the rest of the paper we will refer to the points in
$V$ as \emph{vertices}, to distinguish them from other points in the
plane.

For a given graph $G$ with vertex set $V$, we say that $H$ is
 a $t$-\emph{spanner} of $G$ if, for any pair of vertices
$u,v \in V$, a shortest path in $H$ from $u$ to $v$ is no longer than
$t$ times the length of a shortest path in $G$ between $u$ and $v$. 
A graph $H$ is a $t$-\emph{spanner} of $V$ if $H$ is a $t$-spanner 
of the complete graph on $V$. 
The value $t$ is called the \emph{stretch factor} of $H$. If $t$ is constant, then $H$ is called a
\emph{length spanner}, or simply a \emph{spanner}.

The spanning properties of Yao graphs have been extensively studied. \autoref{tab:yaoresults} summarizes 
some results that are relevant to this paper.
%It has been shown that 
%$Y_2$ and $Y_3$ are not spanners~\cite{Nawar09}, $Y_4$ is a spanner with stretch factor 
%$8 \sqrt{2}(26+23\sqrt{2}) \lesssim 662.16$~\cite{BDD+12}, $Y_5$ is a spanner with stretch factor 
%$2+\sqrt{3} \lesssim 3.74$~\cite{BBD+15}, $Y_6$ is a spanner with stretch factor 5.8~\cite{BBD+15}, 
%and for $k \ge 7$, $Y_k$ is a spanner with stretch factor $\frac{1+\sqrt{2-2\cos\theta}}{2\cos\theta-1}$, 
%where $\theta = 2\pi/k$~\cite{BDD+12arXiv}.  

\begin{table}
\begin{center}
\begin{tabular} {|c|c|c|}
\hline
Reference & Graph & {Stretch Factor}  \\
\hline
\cite{Nawar09} & $Y_2$, $Y_3$ & $\infty$ \\
\hline
\cite{BDD+12} & $Y_4$ & $~8 \sqrt{2}(26+23\sqrt{2}) \lesssim 662.16$~ \\
\hline
\cite{BBD+15} & $Y_5$ & $2+\sqrt{3} \lesssim 3.74$ \\
\hline
\cite{BBD+15} & $Y_6$  & 5.8 \\ 
\hline 
\cite{BDD+12arXiv} & ~$Y_k$, $k \ge 7$~ & ~$(1+\sqrt{2-2\cos\theta})/(2\cos\theta-1)$, \mbox{where } $\theta = 2\pi/k$~\\
\hline
\cite{BK+14} & $~Y_4^\infty$ & 6.31 \\
\hline
\bf{~[this paper]~} & $~Y_4^\infty$ & $\bf{13+8\sqrt{2} \lesssim 4.94}$\\
\hline
\bf{[this paper]} & $Y_4$ & $\bf{(11+7\sqrt{2})\sqrt{4+2\sqrt{2}} \lesssim 54.62}$\\
\hline
\end{tabular}
\end{center}
\caption{Upper bounds on the stretch factor of Yao graphs.}
\label{tab:yaoresults}
\end{table}

\paragraph{Our contributions.} We show that the stretch factor of $Y_4$ is at most 
$(11+7\sqrt{2})\sqrt{4+2\sqrt{2}} \lesssim 54.62$, which is a significant improvement upon the best 
previously known upper bound of $662.16$ from~\cite{BDD+12}.  
We also show that the stretch factor of $Y_4^\infty$ is at most $13+8\sqrt{2} \lesssim 4.931$, 
improving the $6.31$ bound from~\cite{BK+14}. 
The graph $Y_4^\infty$ is of particular interest due to its planarity property (as a subgraph of the $L_\infty$-Delaunay triangulation~\cite{BK+14}) and its applications in scheduling problems~\cite{LW80}.

\section{Definitions}
Let $V$ be a set of vertices in the plane. 
For each vertex $u \in V$, let $x_u$ and $y_u$ denote the $x$-coordinate and the $y$-coordinate of $u$, respectively. For every pair of vertices $u, v \in V$, the \emph{horizontal distance} between $u$ and $v$ is $d_x(u, v) = |x_u -x_v|$; the 
\emph{vertical distance} is $d_y(u, v) = |y_u - y_v|$; the Euclidean distance is $d_2(u, v) = \sqrt{d_x(u, v)^2 + d_y(u, v)^2}$; and 
the $L_\infty$-distance is $d_\infty(u, v) = \max{\{d_x(u, v), d_y(u, v)\}}$. For any plane graph $G$ with vertex set $V$, the \emph{weight} of an edge 
in $G$ is the Euclidean distance between its endpoints; the \emph{length} of a path in $G$ is the sum of the weights of its constituent edges; 
and the distance in $G$ between $u, v \in V$, denoted $d_G(u, v)$, is the length of a shortest path in $G$ between $u$ and $v$. 
We denote by $(u, v)$ the \emph{edge} or the \emph{line segment} connecting $u$ and $v$, and the distinction between 
the two will become clear from the context. 

%
%%%%%%%%%%%%%%%%%%%%%%%%%%%%%%%%%Figure Begin
\begin{figure}[htbp]
\centering
\includegraphics[width=0.3\linewidth]{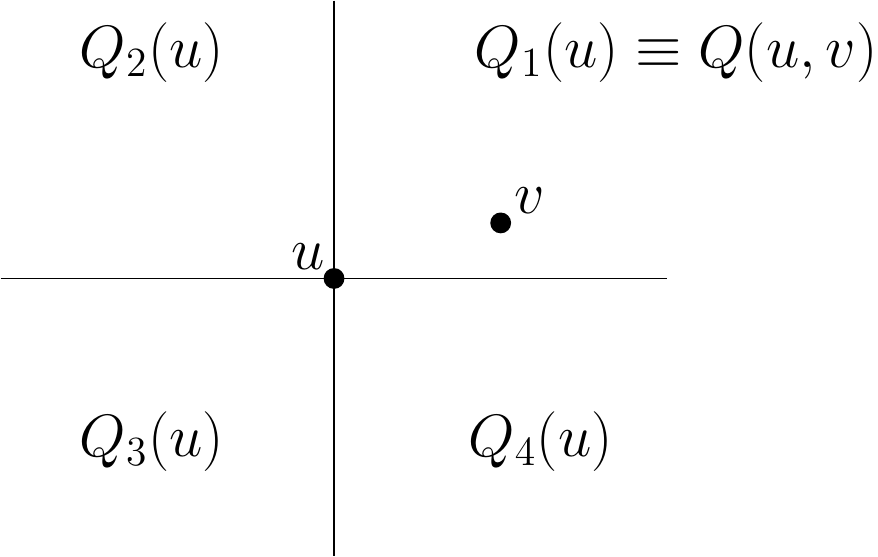}
\caption{(a) Definitions: quadrants $Q_i(u)$, $i = 1, 2, 3, 4$, and $Q(u, v)$.}
\label{fig:quadrants}
\end{figure}
%%%%%%%%%%%%%%%%%%%%%%%%%%%%%%%%%Figure End
%

A \emph{cone} is the region in the plane between two rays that radiate from the same point. With each vertex $u \in V$ we associate four cones of angles $\pi/2$ delimited by two lines parallel to the coordinate axes passing through $u$. We label the cones $Q_{1}(u), Q_{2}(u), Q_3(u)$ and $Q_4(u)$ in counterclockwise order, starting to the first quadrant. Refer to~\autoref{fig:quadrants}. 
%We assume that arithmetic on the cone labels is modulus 4, so that cones $Q_{i-1}$ and $Q_{i+1}$ are well defined. 
To avoid overlapping boundaries, we assume that each cone is half-open and half-closed, meaning that a cone includes its 
clockwise bounding ray but excludes its counterclockwise bounding ray. For any $u, v \in V$, let $Q(u, v)$ denote the quadrant with apex $u$ that contains $v$. 

The directed Yao graph $\arr{Y_4}$ with vertex set $V$ is constructed as follows. For each vertex $u \in V$ and each cone $Q_i(u)$, 
for $i = 1 \ldots 4$, extend a directed edge $\arr{(u,v)}$ from $u$ to a vertex $v \in V$ that lies in $Q_i(u)$ and minimizes the Euclidean distance $d_2(u,v)$. Ties are broken arbitrarily. 
The Yao graph $\arr{Y_4^\infty}$ is defined similarly to $Y_4$, with two differences: (i) it uses the $L_\infty$-distance $d_\infty(u, v)$ rather than the Euclidean distance $d_2(u, v)$, and (ii) ties are broken by selecting the most counterclockwise edge in each quadrant.
The undirected Yao graph $Y_4$ includes all edges of $\arr{Y_4}$ but ignores their directions, and similarly for $Y_4^\infty$. We are interested 
in the stretch factors of $Y_4$ and $Y_4^\infty$. 

Let $Del^\infty$ denote the Delaunay triangulation on $V$ in the $L_\infty$-metric, defined as follows. For any pair of vertices $u, v \in V$, an edge $(u,v)$ is in $Del^\infty$ if and only if there is an axis-aligned square with $u$ and $v$ on its boundary that contains no other vertices in its interior.  A well-known property of $Del^\infty$ is that, for each triangle $T$ in $Del^\infty$, the square whose sides pass through the three vertices of $T$ (the \emph{circumsquare} of $T$) has no vertices of $V$ in its interior. 

For any polygon $P$, let $\partial P$ denote the boundary of $P$. For any two vertices $u$ and $v$, let $R(u, v)$ denote the 
rectangle with sides parallel to the coordinate axes having $u$ and $v$ as opposite corners. (See~\autoref{fig:rect}.)
We say that two edges \emph{intersect} ({\emph{cross}) if they share a point (an interior point). Note that by this definition, two intersecting edges may share an endvertex. Throughout the paper, we use the symbol $\oplus$ to denote the concatenation operator. 

\section{$Y^\infty_4$ in the $L_\infty$ Metric}
\label{sec:y4inf}
In this section we show that $Y_4^\infty$ has stretch factor at most $\sqrt{13+8\sqrt{2}} \lesssim 4.931$. This
 improves upon the best previously known stretch factor of $(1+\sqrt{2})\sqrt{4+2\sqrt{2}} \lesssim 6.31$ from~\cite{BK+14}. 
We begin with the following result established in~\cite{BK+14}.
\begin{lemma}
The graph $Y_4^\infty$ is a subgraph of $Del^\infty$, a $(1+\sqrt{2})$-spanner of $Del^\infty$ and also a $(1+\sqrt{2})\sqrt{4+2\sqrt{2}}$-spanner of $V$. 
\label{lem:bk14}
\end{lemma}
Although not explicitly stated, the proof of~\autoref{lem:bk14} from~\cite{BK+14} implies the following result.  

\begin{lemma}
 For each triangle $\triangle{uvw} \in Del^\infty$, at least two of its edges are 
in $Y_4^\infty$. If $(u,v)$ is not in $Y_4^\infty$, then $u$ and $v$ lie on opposite sides of the circumsquare of $\triangle{uvw}$. 
\label{lem:quad}
\end{lemma}
An immediate consequence of~\autoref{lem:quad} is the following.

\begin{corollary}
\label{cor:virtual} For each triangle $\triangle{uvw} \in Del^\infty$, if $(u,v)$ is not in $Y_4^\infty$, then 
\[
d_\infty(u, v) \ge \max \{d_\infty(u, w), d_\infty(w, v)\}.
\]
\end{corollary}
These together yield the following result. 

\begin{lemma}
For each edge $(u,v) \in Del^\infty$, there is a path in $Y_4^\infty$ of length
\[
    d_{Y_4^\infty} (u, v) \le (1+\sqrt{2}) \cdot d_\infty(u, v)
\]
\label{lem:bonus}
\end{lemma}
\begin{proof}
If $(u,v)$ is in $Y_4^\infty$, then the theorem clearly holds. So assume that $(u,v) \not\in Y_4^\infty$. 
Let $T = \triangle uvw$ be a triangle in $Del^\infty$ with side $(u,v)$. By~\autoref{lem:quad}, both $(u,w)$ and $(v,w)$ are in $Y_4^\infty$. 
Thus $(u,w) \oplus (w,v)$ is a path in $Y_4^\infty$ between $u$ and $v$ of length $d_2(u,w) + d_2(w,v)$.  Also by~\autoref{lem:quad}, 
$u$ and $v$ lie on opposite sides of $T$'s circumsquare. This implies that $d_2(u,w) + d_2(w,v)$ is bounded above by 
$(1+\sqrt{2})d_\infty(u, v)$, which is achieved when one of $(u,w)$ and $(v,w)$ is a side, and the other is a diagonal of $T$'s circumsquare. 
\end{proof}

\noindent
%The main result of this section, showing that the stretch factor of $Y_4^\infty$ is bounded above by $\sqrt{13+8\sqrt{2}} \lesssim 4.931$, follows from the following key theorem: 
The following theorem is key in establishing an upper bound on the stretch factor of $Y_4^\infty$. 
\begin{theorem}
\label{thm:key}
Let $a$ and $b$ be arbitrary vertices in $V$.  If $x = d_\infty(a, b) = \max\{d_x(a, b), d_y(a, b)\}$ and 
$y = \min\{d_x(a, b), d_y(a, b)\}$, then 
\[
d_{Y_4^\infty} (a, b) \le 2(1+\sqrt{2})x+y
\]
\end{theorem}

\noindent
We delay the proof of~\autoref{thm:key} until we establish some essential ingredients. The main result of this section, stated in~\autoref{thm:main} below, is an immediate consequence of~\autoref{thm:key}.
\begin{theorem}
The stretch factor of $Y_4^\infty$ on a set of points $V$ is at most 
\[\sqrt{13+8\sqrt{2}} \lesssim 4.931\]
\label{thm:main}
\end{theorem}
\begin{proof}
By~\autoref{thm:key}, the stretch factor of $Y_4^\infty$ is no greater than the maximum of the function 
\[
\frac{2(1+\sqrt{2})x+y}{\sqrt{x^2+y^2}}
\]
which is equal to $\sqrt{13+8\sqrt{2}}$ when $x/y = 2(1+\sqrt{2})$. 
\end{proof}

%
%%%%%%%%%%%%%%%%%%%%%%%%%%%%%%%%%Figure Begin
\begin{figure}[hpt]
\centering
\includegraphics[width=0.99\linewidth]{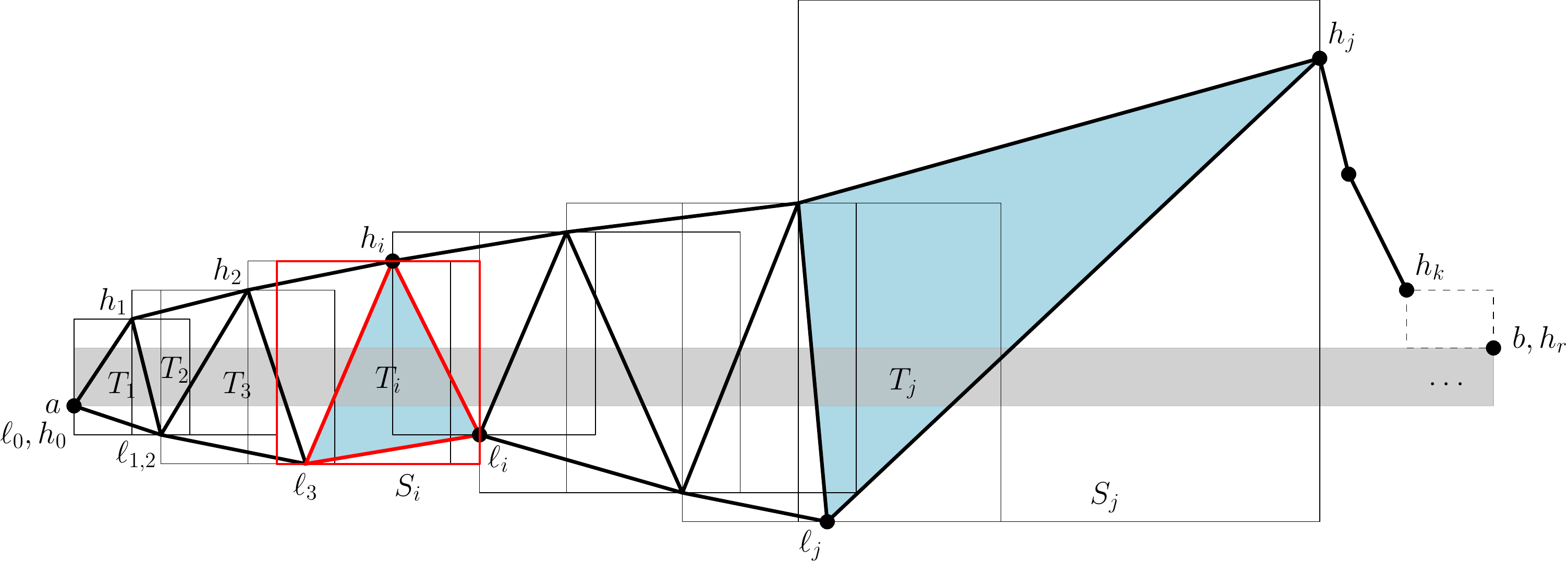}
\caption{\autoref{lem:key}: Squares $S_1, S_2, \ldots, S_{j-1}$ are not inductive. Square $S_j$ is inductive. Vertex $\ell_i$ is on the east side of square $S_i$.}
\label{fig:promising}
\end{figure}
%%%%%%%%%%%%%%%%%%%%%%%%%%%%%%%%%Figure End
%
Our approach in proving~\autoref{thm:key} mimics the approach 
used in~\cite{BG+12} to establish a stretch factor of $\sqrt{4+2\sqrt{2}}$ for $Del^\infty$. Before describing this 
approach, we need to introduce some definitions. To make it easy for the interested reader, most of the terminology 
in this section is similar to the one used in~\cite{BG+12}. We assume without loss of generality that 
$a$ has coordinates $(0, 0)$. In this case, the definitions used in the statement of~\autoref{thm:key} 
imply that $b$ has coordinates $(x, y)$. 
Let $T_1, T_2, \ldots, T_r$ be the sequence of triangles in $Del^\infty$ that intersect the line segment $ab$ when moving 
from $a$ to $b$. For each triangle $T_i$, let $(h_i, \ell_i)$ be the rightmost edge of $T_i$ that intersects $ab$, with $h_i$ above 
$ab$ and $\ell_i$ below $ab$. We also let $h_0 = \ell_0 = a$, $h_r = b$ and $\ell_{r-1} = \ell_r$. 
Note that some vertices coincide: either $h_i = h_{i-1}$ and $T_i = \triangle{h_i\ell_{i-1}\ell_{i}}$, or 
$\ell_i = \ell_{i-1}$ and $T_i = \triangle{h_{i-1}h_{i}\ell_i}$. Let $S_i$ be the circumsquare of $T_i$. 
We call the square $S_i$ \emph{inductive} if $d_\infty(h_i, \ell_i) = d_x(h_i, \ell_i)$. The vertex $h_i$ or $\ell_i$ with the 
larger $x$-coordinate is the \emph{inductive point} of $S_i$. In~\autoref{fig:promising} for example, $h_j$ is the inductive point of $S_j$. 

One key ingredient in proving~\autoref{thm:key} is the following lemma. 
\begin{lemma}
Assume that $R(a, b)$ is empty. If no square $S_1, \ldots, S_r$ is inductive, then
\[d_{Y_4^\infty} (a, b) \le 2(1+\sqrt{2})x + y .
\]
Otherwise, let $S_j$ be the first inductive square in the sequence $S_1, \ldots, S_r$. If $h_j$ is the inductive point of 
$S_j$, then 
\[d_{Y_4^\infty} (a, h_j) +(y_{h_j}-y) \le 2(1+\sqrt{2})x_{h_j} .
\]
If $\ell_j$ is the inductive point of $S_j$, then 
\[d_{Y_4^\infty} (a, \ell_j) - y_{\ell_j} \le 2(1+\sqrt{2})x_{\ell_j} .
\]
\label{lem:key}
\end{lemma}
\begin{proof}
Because $S_j$ is inductive, $d_\infty(\ell_j, h_j) = |x_{h_j} - x_{\ell_j}|$ (by definition). This along with~\autoref{lem:bonus} implies 
\begin{equation}
d_{Y_4^\infty}(\ell_j, h_j) \le (1+\sqrt{2})|x_{h_j} - x_{\ell_j}|
\label{eq:lh}
\end{equation}
Assume first that $h_j$ is the inductive point of $S_j$, meaning that $x_{h_j} > x_{\ell_j}$. In this case 
$h_j$ lies on the east side of $S_j$ and $\ell_j$ lies on the west or south side of $S_j$. 
Let $T_i$ be the first triangle encountered when moving from $T_j$ leftward toward $T_1$, 
such that either $i = 0$ or $\ell_i$ lies on the east side of $T_i$. Refer to~\autoref{fig:promising}. Note that $T_i \neq T_j$, since $\ell_j$ does not lie on the east side of $T_j$. Then all edges in $Del^\infty$ on the path $p_{ij} = \ell_i, \ell_{i+1}, \ldots \ell_j$ 
span between the west and south sides of their enclosing square, and by~\autoref{lem:quad} they are also in $Y_4^\infty$. Also note that the path 
$p_{ij}$ descends vertically, therefore $y_{\ell_i} > y_{\ell_j}$. These together with the triangle inequality 
applied on each edge of $p_{ij}$ imply
\begin{equation}
d_{Y_4^\infty}(\ell_i, \ell_j) < (x_{\ell_j} - x_{\ell_i}) + (y_{\ell_i}-y_{\ell_j}) .
\label{eq:ll}
\end{equation}
We now use the combined results from Lemmas 9 and 11 from~\cite{BG+12} showing that 
\[
d_{Del^\infty} (a, \ell_i) \le 2x_{\ell_i} .
\]
This along with the fact that $Y_4^\infty$ is a $(1+\sqrt{2})$-spanner of $Del^\infty$ implies that 
\begin{equation}
d_{Y_4^\infty} (a, \ell_i) \le 2(1+\sqrt{2})x_{\ell_i} .
\label{eq:al}
\end{equation}
We are now ready to evaluate
\begin{eqnarray*}
d_{Y_4^\infty}(a, h_j)+(y_{h_j}-y) & \le & d_{Y_4^\infty} (a, \ell_i) + d_{Y_4^\infty} (\ell_i, \ell_j) + d_{Y_4^\infty} (\ell_j, h_j) + y_{h_j}
\end{eqnarray*}
Substituting inequalities~(\ref{eq:lh}),~(\ref{eq:ll}) and~(\ref{eq:al}) in the right hand side above yields
\begin{eqnarray}
\nonumber d_{Y_4^\infty}(a, h_j)+(y_{h_j}-y)  & < & 2(1+\sqrt{2})x_{\ell_i} + (x_{\ell_j} - x_{\ell_i}) + (y_{\ell_i}-y_{\ell_j}) + (1+\sqrt{2})(x_{h_j} - x_{\ell_j}) + y_{h_j}\\
& < & (2 + 2\sqrt{2}-1)x_{\ell_i} - (1+\sqrt{2}) x_{\ell_j}+ (x_{\ell_j} + y_{h_j} - y_{\ell_j}) + (1+\sqrt{2})x_{h_j}
\label{eq:hj1}
\end{eqnarray}
We safely ignored the quantity $y_{\ell_i} < 0$ in the right hand side of the inequality above. Recall that 
$d_\infty(\ell_j, h_j) = x_{h_j} - x_{\ell_j}$ (since $S_j$ is inductive), therefore $x_{\ell_j} + y_{h_j} - y_{\ell_j} \le x_{h_j}$. 
Also note that $x_{\ell_j} > x_{\ell_i} \ge 0$, therefore $-x_{\ell_j} < -x_{\ell_i}$. Substituting these inequalities in~(\ref{eq:hj1})
yields
\begin{eqnarray*}
d_{Y_4^\infty}(a, h_j)+(y_{h_j}-y) & < & (1+2\sqrt{2}-1-\sqrt{2})x_{\ell_i} + x_{h_j} + (1+\sqrt{2})x_{h_j}\\
&\le & \sqrt{2}x_{\ell_i} + x_{h_j} + (1+\sqrt{2})x_{h_j} \\
& < & 2(1+\sqrt{2})x_{h_j}
\end{eqnarray*}
This latter inequality follows from the fact that $x_{\ell_i} < x_{h_j}$. 

%
%%%%%%%%%%%%%%%%%%%%%%%%%%%%%%%%%Figure Begin
\begin{figure}[hpt]
\centering
\includegraphics[width=0.99\linewidth]{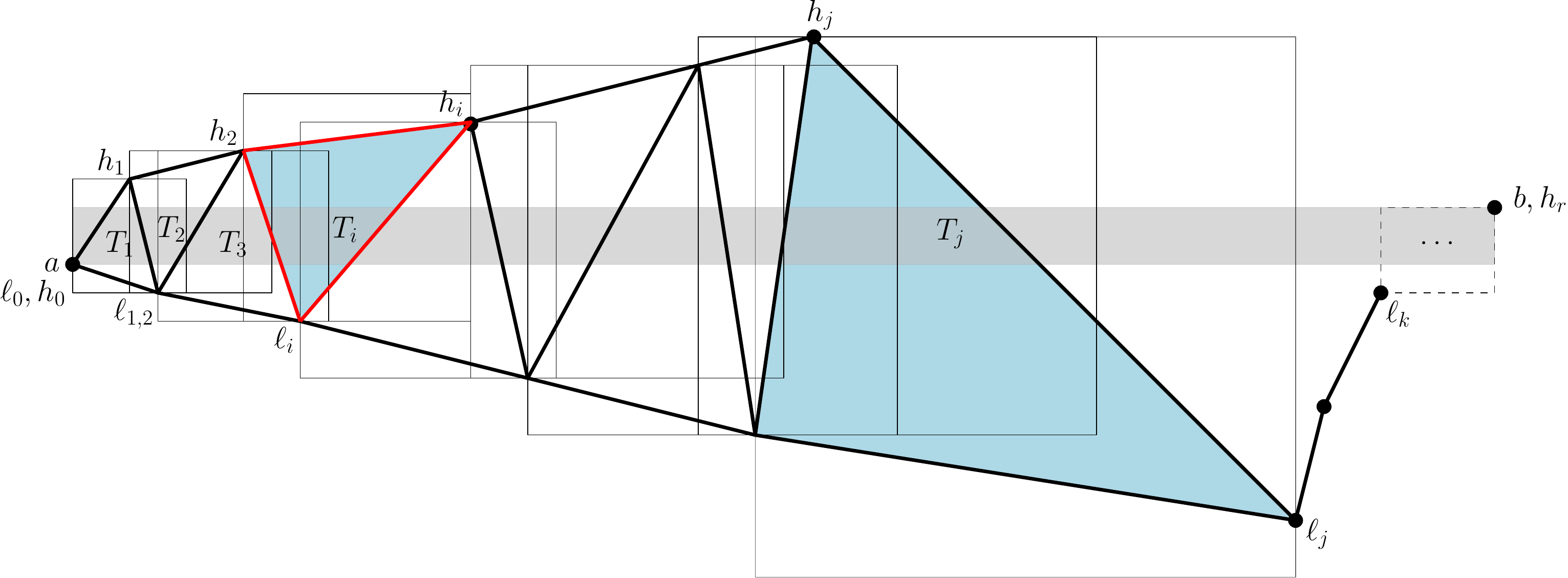}
\caption{\autoref{lem:key}: Squares $S_1, S_2, \ldots, S_{j-1}$ are not inductive. Square $S_j$ is inductive. Vertex $h_i$ is on the east side of square $S_i$.}
\label{fig:promising2}
\end{figure}
%%%%%%%%%%%%%%%%%%%%%%%%%%%%%%%%%Figure End
%

Assume now that $\ell_j$ is the inductive point of $S_j$, so $\ell_j$ lies on
the east side of $S_j$ and $h_j$ lies on the west or north side of $S_j$. The analysis for this case is symmetric to the one used for the previous case. 
Redefine $T_i$ to be the first triangle encountered when moving from $T_j$ leftward toward $T_1$, 
such that either $i = 0$ or $h_i$ lies on the east side of $T_i$. Refer to~\autoref{fig:promising2}. 
Arguments similar to the ones used for the previous case show that 
\begin{eqnarray*}
d_{Y_4^\infty}(a, h_i)- y_{\ell_j} & \le & d_{Y_4^\infty} (a, h_i) + d_{Y_4^\infty} (h_i, h_j) + d_{Y_4^\infty} (h_j, \ell_j) - y_{\ell_j}  \\
& \le & 2(1+\sqrt{2})x_{h_i} + (x_{h_j} - x_{h_i}) + (y_{h_j}-y_{h_i}) + (1+\sqrt{2})(x_{\ell_j} - x_{h_j}) - y_{\ell_j}\\
& < & (2+2\sqrt{2}-1)x_{h_i} - (1+\sqrt{2})x_{h_j} + (x_{h_j} + y_{h_j}-y_{h_i}) + (1+\sqrt{2})x_{\ell_j}\\
& \le & (1+2\sqrt{2}-1-\sqrt{2})x_{h_j} + x_{\ell_j} + (1+\sqrt{2})x_{\ell_j}\\
&\le & \sqrt{2} x_{h_j} + x_{\ell_j} +(1+\sqrt{2})x_{\ell_j}\\
& < & 2(1+\sqrt{2})x_{\ell_j}
\end{eqnarray*}
In deriving these inequalities we ignored the term $-y_{\ell_j} < 0$ and used the fact that $x_{h_i} < x_{h_j} < x_{\ell_j}$ and 
$x_{h_j} + y_{h_j}-y_{h_i} < x_{h_j}  + y_{h_j} - y_{\ell_j} \le x_{\ell_j}$ (by the lemma statement that $S_j$ is inductive). 
\end{proof}

\noindent
We are now ready to prove~\autoref{thm:key}. Our proof follows closely the proof from~\cite{BG+12} used to establish a similar result in the context of $Del^\infty$, with some changes necessary to handle edges in $Del^\infty$ that do not exist in $Y_4^\infty$. 

\medskip
\noindent
\textbf{Theorem~\ref{thm:key}}\emph{
Let $a$ and $b$ be arbitrary vertices in $V$.  If $x = d_\infty(a, b) = \max\{d_x(a, b), d_y(a, b)\}$ and 
$y = \min\{d_x(a, b), d_y(a, b)\}$, then} 
\[d_{Y_4^\infty} (a, b) \le 2(1+\sqrt{2})x+y\]
\begin{proof}
%\paragraph{Proof of~\autoref{thm:key}.} 
By the theorem statement,  $a$ and $b$ are two arbitrary points in $V$ of coordinates $(0,0)$ and $(x, y)$ respectively, with 
$x = d_\infty(a, b) \ge y$. Our goal is to prove that  $d_{Y_4^\infty} (a, b) \le 2(1+\sqrt{2})x+y$. 
The proof is by induction on the $L_\infty$-distance between pairs of points in $V$. 

For the base case, assume that $a$ and $b$ are a closest pair of vertices in the $L_\infty$-metric. In this case $ab \in Del^\infty$ and, 
by~\autoref{lem:bonus}, $d_{Y_4^\infty} (a, b) \le (1+\sqrt{2})\cdot d_\infty(a, b) = (1+\sqrt{2})x$. Thus the theorem holds for the base case. 

For the induction step, assume that $a, b \in V$ are arbitrary, and that the theorem holds for all pairs of vertices in $V$ strictly closer than $d_\infty(a, b)$ in 
the $L_\infty$-metric. We discuss two cases, depending on whether the interior of $R(a,b)$ is empty or not. 

%
%%%%%%%%%%%%%%%%%%%%%%%%%%%%%%%%%Figure Begin
\begin{figure}[hpt]
\centering
\includegraphics[width=0.98\linewidth]{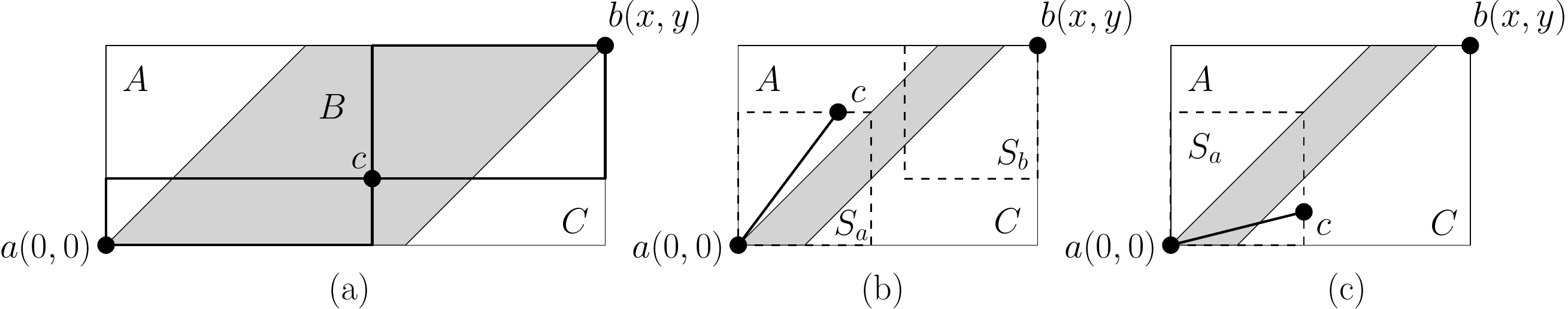}
\caption{\autoref{thm:key}: $R(a, b)$ is non-empty.}
\label{fig:rect}
\end{figure}
%%%%%%%%%%%%%%%%%%%%%%%%%%%%%%%%%Figure End
%

\noindent {\bf Case 1.} 
Assume first that the interior of $R(a,b)$ is not empty. Partition the interior of $R(a, b)$ into three regions (call them $A$, $B$ and $C$ left to right) with two lines of slope one passing through $a$ and $b$. Any point $c$ in the mid-region $B$ (shaded in~\autoref{fig:rect}) satisfies $x_c \ge y_c$ 
and $x-x_c \ge y-y_c$. If there is such a point, then we can apply induction on the vertex pairs $(a, c)$ and $(c, b)$ to obtain  
$d_{Y_4^\infty}(a, c) \le 2(1+\sqrt{2})x_c + y_c$ and 
$d_{Y_4^\infty}(c, b) \le  2(1+\sqrt{2})(x-x_c) + (y-y_c) $. 
Summing up these two inequalities yields $d_{Y_4^\infty}(a, b) \le d_{Y_4^\infty}(a, c) + d_{Y_4^\infty}(c, b) \le 2(1+\sqrt{2})x+ y$, 
so the theorem holds for this case. 

Let $S_a$ ($S_b$) be the largest \emph{empty} square with bottom left corner $a$ (top right corner $b$) that fits inside $R(a, b)$.  
If region $B$ is empty, then there must be a vertex $c \in V$, with $c \not\in \{a, b\}$, that lies on the boundary of either $S_a$ or $S_b$. 
Assume without loss of generality that there is such a vertex on the boundary of $S_a$, and let $c$ be the most 
counterclockwise such vertex (relative to $a$). In this case $(a,c) \in Y_4^\infty$ (by definition) and therefore 
\begin{equation}
d_{Y_4^\infty} (a, c) = d_2(a, c) < x_c + y_c 
\label{eq:ac}
\end{equation}
If $c$ lies in region $A$ (as in~\autoref{fig:rect}b), then $y_c > x_c$ and $x - x_c > y - y_c$. We apply induction on the vertex pair $(c, b)$ to derive $d_{Y_4^\infty} (c, b) \le 2(1+\sqrt{2})(x-x_c) + (y-y_c) $. This along with~(\ref{eq:ac}) yields  
  \begin{eqnarray*}
d_{Y_4^\infty} (a, b) & \le & d_{Y_4^\infty} (a, c) + d_{Y_4^\infty} (c, b) \\
  & < & x_c + y_c + 2(1+\sqrt{2})(x-x_c) + y-y_c \\
  & \le & 2(1+\sqrt{2})x + y - (1+2\sqrt{2})x_c \\
  & < & 2(1 + \sqrt{2})x + y
  \end{eqnarray*} 
If $c$ lies in region $C$ (as in~\autoref{fig:rect}c), then $x_c > y_c$ and $y - y_c > x - x_c$. We apply induction on the vertex pair $(c, b)$ to derive $d_{Y_4^\infty} (c, b) \le 2(1+\sqrt{2})(y-y_c) + (x-x_c) $. This along with~(\ref{eq:ac}) yields  
   \begin{eqnarray*}
   d_{Y_4^\infty} (a, b) & \le & d_{Y_4^\infty} (a, c) + d_{Y_4^\infty} (c,b)\\
   & \le & x_c + y_c + 2(1+\sqrt{2})(y-y_c) + x-x_c \\
   & \le & 2(1+\sqrt{2})y + x - (1+2\sqrt{2})y_c \\
   & < & (1 + 2\sqrt{2})y + x + y \\
   & < & 2(1+\sqrt{2})x + y
   \end{eqnarray*} 
This latter inequality follows immediately from the fact that $y < x$.   

\medskip
\noindent {\bf Case 2.} 
Assume now that the interior of $R(a,b)$ is empty. 
If no square  in the sequence $S_1, S_2, \ldots, S_r$ is inductive, then by~\autoref{lem:key} we have $d_{Y_4^\infty}(a,b) \le 2(1+\sqrt{2})x+y$ and the theorem holds. Otherwise, let $ S_j$ be the first inductive square in the sequence $S_1, S_2, \ldots, S_r$. 

Assume first that $h_j$ is the inductive point of $S_j$ (so $h_j$ lies on the east side of $S_j$). Let $h_k$ be the first vertex in the sequence $h_j, h_{j+1}, \ldots h_r=b$ such that 
\begin{equation}
x-x_{h_k} \ge y_{h_k}-y > 0
\label{eq:hjk}
\end{equation}
Refer to~\autoref{fig:promising}. Inequality~(\ref{eq:hjk}) implies that $h_k$ is closer to $b$ in the $L_\infty$ metric than $a$ is. This enables us to use induction to determine an upper bound on $d_{Y_4^\infty} (h_k, b)$. Before we do so, note that each edge $(h_p,h_{p+1})$, for any $j \le p < k$, has its endpoints on the north and east sides of its enclosing squares $S_{p+1}$. (The only other alternatives would be for $(h_p,h_{p+1})$ to span between the west and north sides, or between the west and east sides of $S_{p+1}$. In each of these cases $S_{p+1}$, which must pass through a vertex $\ell_{p+1}$ below $b$, would extend too far to the right and include the endpoint $b$, since the horizontal distance from $h_p$ to $b$ is no longer than the vertical distance from $h_p$ to $b$. This contradicts the fact that  $S_{p+1}$ is empty.)  This further implies that the path $h_j, h_{j+1}, \ldots, h_k$ is in $Y_4^\infty$ (by~\autoref{lem:quad}). This along with the triangle inequality applied on each edge along this path yields 
\[
p_{Y_4^\infty}(h_j, h_k) < (x_{h_k}-x_{h_j}) + (y_{h_j}-y_{h_k})
\]
%$ and $x-x_{h_k} \ge y_{h_k}-y \ge 0.$ Since $h_k$ is closer to b, we can apply induction to bound $ d_{Y_4^\infty}(h_k,b)$. 
This observation together with~\autoref{lem:key} used to bound $d_{Y_4^\infty} (a, h_j)$ and the inductive hypothesis used to bound 
$ d_{Y_4^\infty} (h_k, b)$ yields 
 \begin{eqnarray*}
 d_{Y_4^\infty} (a, b)  & \le & d_{Y_4^\infty} (a, h_j) + d_{Y_4^\infty} (h_j, h_k) + d_{Y_4^\infty} (h_k, b) \\
 & < & 2(1+\sqrt{2}) x_{h_j} - (y_{h_j}-y ) + (x_{h_k}-x_{h_j}) + (y_{h_j}-y_{h_k}) + 2(1+\sqrt{2})(x-x_{h_k}) + (y_{h_k}-y) \\
 & = & 2(1+\sqrt{2})x + (2(1+\sqrt{2}) -1) x_{h_j} + (1-2(1+\sqrt{2}))x_{h_k} \\
 & = & 2(1+\sqrt{2})x + (1+2\sqrt{2}) x_{h_j} - (1+2\sqrt{2})x_{h_k}\\
 & \le & 2(1+\sqrt{2})x
 \end{eqnarray*}
The last inequality follows immediately from the fact that $x_{h_j} \le x_{h_k}$. Assume now that $\ell_j$ is the inductive point of $S_j$ (so $\ell_j$ lies on the east side of $S_j$). Let $\ell_k$ be the first vertex in the sequence $\ell_j, \ell_{j+1}, \ldots \ell_r$ such that 
\begin{equation*}
x-x_{\ell_k} \ge y - y_{\ell_k} > 0
\label{eq:hjk}
\end{equation*}
Refer to~\autoref{fig:promising2}. Arguments similar to the ones used in the previous case show that 
\[
p_{Y_4^\infty}(\ell_j, \ell_k) < (x_{\ell_k}-x_{\ell_j}) + (y_{\ell_k}-y_{\ell_j})
\]
This along with~\autoref{lem:key} used to bound $d_{Y_4^\infty} (a, \ell_j)$ and the inductive hypothesis used to bound 
$d_{Y_4^\infty} (\ell_k, b)$ yields 
 \begin{eqnarray*}
 d_{Y_4^\infty} (a,b)  & \le & d_{Y_4^\infty} (a, \ell_j) + d_{Y_4^\infty} (\ell_j, \ell_k) + d_{Y_4^\infty} (\ell_k, b) \\
 & < & 2(1+\sqrt{2}) x_{\ell_j} + y_{\ell_j} + (x_{\ell_k}-x_{\ell_j}) + (y_{\ell_k}-y_{\ell_j}) + 2(1+\sqrt{2})(x-x_{\ell_k}) + (y-y_{\ell_k}) \\
 &= & 2(1+\sqrt{2})x + y + (2(1+\sqrt{2}) -1) x_{\ell_j} + (1-2(1+\sqrt{2}))x_{\ell_k} \\
 & = & 2(1+\sqrt{2})x + y + (1+2\sqrt{2}) x_{\ell_j} - (1+2\sqrt{2})x_{\ell_k}\\
 & \le & 2(1+\sqrt{2})x + y
 \end{eqnarray*}
This concludes the proof of~\autoref{thm:key}. 
\end{proof}

\section{$Y_4$ in the $L_2$ Metric}
\label{sec:y4}
In this section we turn to the Yao graph $Y_4$ defined in the Euclidean metric space. 
It has been shown that, corresponding to each edge $(a,b) \in Y_4^\infty$, there is a path in $Y_4$ of length 
$d_{Y_4}(a, b) \le (26+23\sqrt{2})\cdot d_2(a,b)$ (Lemma 9 from~\cite{BDD+12}). Combined with the result of~\autoref{thm:main}, which shows that 
$Y_4^\infty$ is a $\sqrt{13+8\sqrt{2}}$-spanner, this yields a stretch factor of $(26+23\sqrt{2})\sqrt{13+8\sqrt{2}} \lesssim 288.59$ for $Y_4$. This improves upon the best currently known stretch factor of $8 \sqrt{2}(26+23\sqrt{2}) \lesssim 662.16$ for $Y_4$ established in~\cite{BDD+12}. In this section we further reduce the stretch factor of $Y_4$ to $(11+7\sqrt{2})\sqrt{4+2\sqrt{2}} \lesssim 54.62$. 

%
%%%%%%%%%%%%%%%%%%%%%%%%%%%%%%%%%Figure Begin
\begin{figure}[htbp]
\centering
\begin{tabular}{c@{\hspace{0.15\linewidth}}c}
\includegraphics[width=0.2\linewidth]{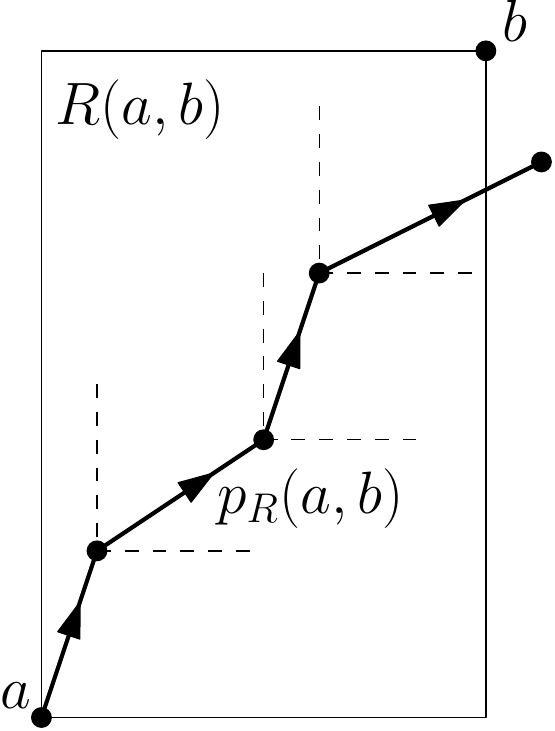} &
\includegraphics[width=0.25\linewidth]{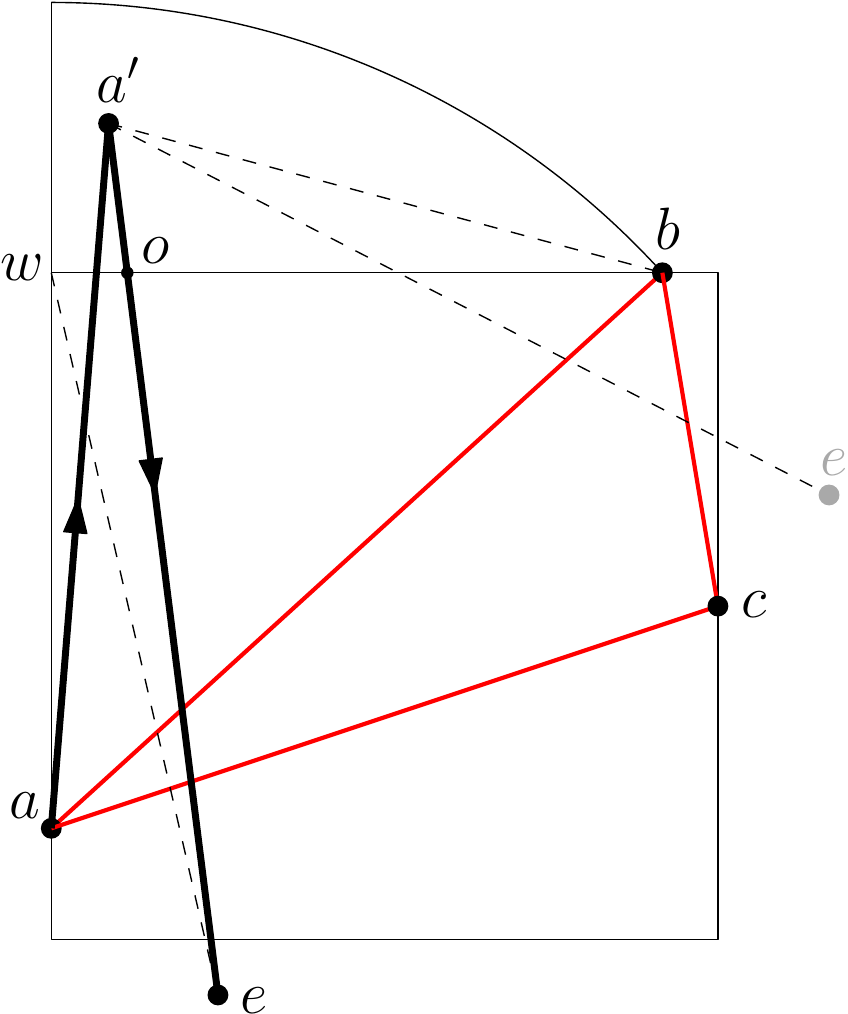} \\
(a) & (b) 
\end{tabular}
\caption{(a) Greedy path $p_R(a, b)$ (b) If $(a',e)$ crosses $(a,b)$, then $\arr{(a',e)} \notin Y_4$.}
\label{fig:y4defs}
\end{figure}
%%%%%%%%%%%%%%%%%%%%%%%%%%%%%%%%%Figure End
%

For ease of presentation, we introduce a few definitions. 
Let $p_R(a, b)$ denote the greedy path that begins at $a$, follows the $Y_4$ edges pointing in the direction of $b$, and ends at the first vertex exterior to, or on the boundary of, $R(a,b)$.~\autoref{fig:y4defs}a illustrates this definition. 
Let $d_R(a, b)$ denote the length of $p_R(a, b)$.  
In our proofs we use the following preliminary results from~\cite{BDD+12}. 

\begin{proposition}[\cite{BDD+12}]
For any triangle $\triangle abc$, $d_2(a, c)^2 < d_2(a, b)^2 + d_2(b, c)^2$, if $\angle {bac} < \pi/2$.
\label{prop:coslaw}
\end{proposition}

\begin{lemma}[\cite{BDD+12}]
$d_R(a, b) \le d_2(a,b)\sqrt{2}$, and each edge on $p_R(a, b)$ is no longer than $(a,b)$. 
\label{lem:R}
\end{lemma}

\begin{lemma}
Let $\arr{(a,b)}$  and $\arr{(c,d)}$ be two edges in $Y_4$ that intersect. % (i.e., they share an interior point or an endvertex). 
If $\arr{(a,b)}$ and $\arr{(c,d)}$ share an interior point, let $(x, y)$ be a shortest side of the quadrilateral with vertices $a$, $b$, $c$ and $d$; otherwise, let $x = y$ be the common endpoint. 
In either case, 
\[
d_{Y_4}(x,y) \le 3(2+\sqrt{2}) \cdot \max\{d_2(a,b), d_2(c,d)\}.
\]
\label{lem:y4cross}
\end{lemma}
\begin{proof}
This result follows immediately from two intermediate results established in~\cite{BDD+12}. 
If $x = y$, then $d_{Y_4}(x,y) = 0$ and the lemma clearly holds. Otherwise, 
Lemma 4 from~\cite{BDD+12} shows that 
\[d_2(x, y) \le \max\{d_2(a,b), d_2(c,d)\}/\sqrt{2}\]
Lemma 8 from~\cite{BDD+12} shows that $d_{Y_4}(x, y) \le \frac{6}{\sqrt{2}-1}\cdot d_2(x, y)$. 
These together yield the inequality stated by the lemma. 
\end{proof}

\noindent
We need one more lemma before we turn to the main result of this section. 
\begin{lemma}
Let $\triangle abc \in Del^\infty$ and let  $S$ be its circumsquare. Assume that $(a, b) \notin Y_4$ and $a, b$ lie on adjacent sides that meet at corner $w$ of  $S$. Let $\arr{(a,a')} \in Y_4$ and $\arr{(a',e)} \in Y_4$ be such that $a' \in Q(a, b)$ and $e \in Q(a', b)$. If $(a,a')$ crosses the line segment $(w,b)$, then $(a',e)$ may not cross $(a,b)$. 
\label{lem:y4zigzag}
\end{lemma}
\begin{proof}
Assume to the contrary that $(a,a')$ crosses $(w,b)$ and $(a',e)$ crosses $(a,b)$. Refer to~\autoref{fig:y4defs}b. 
By definition $S$ is empty of vertices, therefore both $a'$ and $e$ lie outside of $S$.  
%Observe first that $(a',e)$ must cross the side of $S$ opposite to $(u, b)$; otherwise, $\angle{a'be}$ would be 
%obtuse and $d_2(a', e) > d_2(a', b)$, contradicting the fact that $\arr{(a', e)} \in Y_4$. 
It follows that $(w,e)$ is longer than the side length of $S$, so the inequality $d_2(w, e) > d_2(w, b)$ holds. 
Let $o$ be the intersection point between $(w, b)$ and $(a', e)$. Summing up the triangle inequalities for  
$\triangle woe$ and $\triangle a'ob$ yields $d_2(w, e) + d_2(a', b) < d_2(w, b) + d_2(a', e)$. This along with $d_2(w, e) > d_2(w, b)$ yields $d_2(a', b) < d_2(a', e)$, contradicting the fact that $\arr{(a',e)} \in Y_4$. It follows that $(a', e)$ may not cross $(a, b)$ and the lemma holds. 
\end{proof}

We are now ready to establish the main result of this section, showing that there is a short path in $Y_4$ between the endpoints of 
each edge in $Del^\infty$. 
\begin{theorem}
For each edge $(a, b) \in Del^\infty$, $d_{Y_4}(a, b) \le (11 + 7\sqrt{2}) \cdot d_2(a,b)$.
\label{thm:y4main}
\end{theorem}
\begin{proof}
If $(a, b) \in Y_4$, then $d_{Y_4}(a, b) = d_2(a, b)$ and the theorem holds. So assume that 
$(a, b) \notin Y_4$, and let $\arr{(a, a')} \in Y_4$, with $a' \in Q(a, b)$. By definition, 
\begin{equation}
d_2(a, a') \le d_2(a, b) 
\label{eq:aa'}
\end{equation}
This along with~\autoref{prop:coslaw} implies that $d_2(a', b)^2 < d_2(a,a')^2 + d_2(a, b)^2 \le 2\cdot d_2(a, b)^2$, so
\begin{equation}
d_2(a', b) < \sqrt{2} \cdot d_2(a, b)
\label{eq:a'b}
\end{equation}
Since $(a, b) \in Del^\infty$, there is a triangle $\triangle abc \in Del^\infty$ whose circumsquare $S$ contains 
no vertices in its interior.  
We discuss two cases, depending on whether $a$ and $b$ lie on adjacent sides or on opposite sides of $S$. 
%
%%%%%%%%%%%%%%%%%%%%%%%%%%%%%%%%%Figure Begin
\begin{figure}[htbp]
\centering
\begin{tabular}{ccc}
\includegraphics[width=0.27\linewidth]{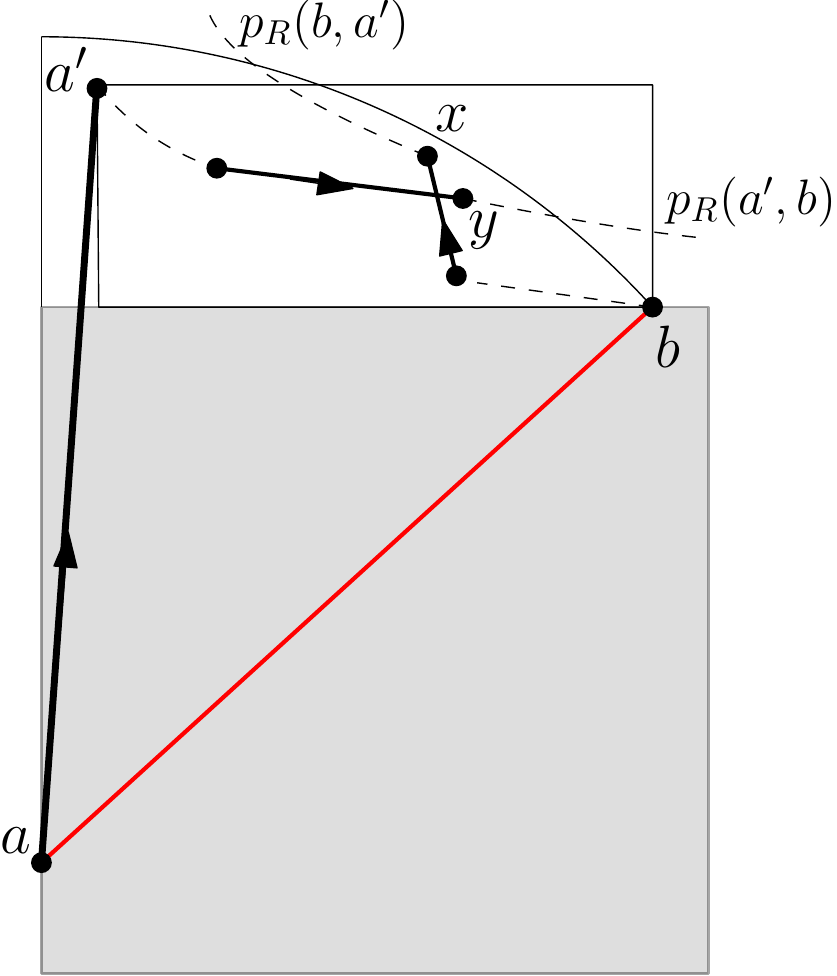} & 
\includegraphics[width=0.32\linewidth]{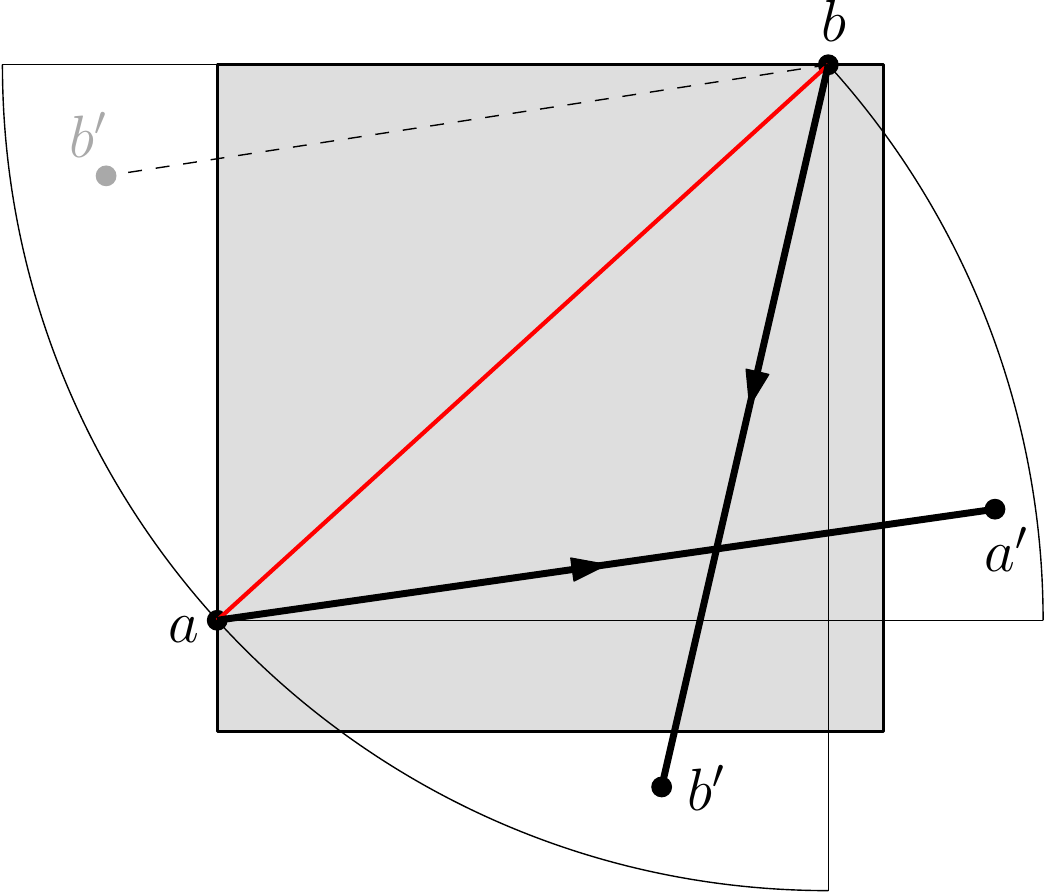} &
\includegraphics[width=0.3\linewidth]{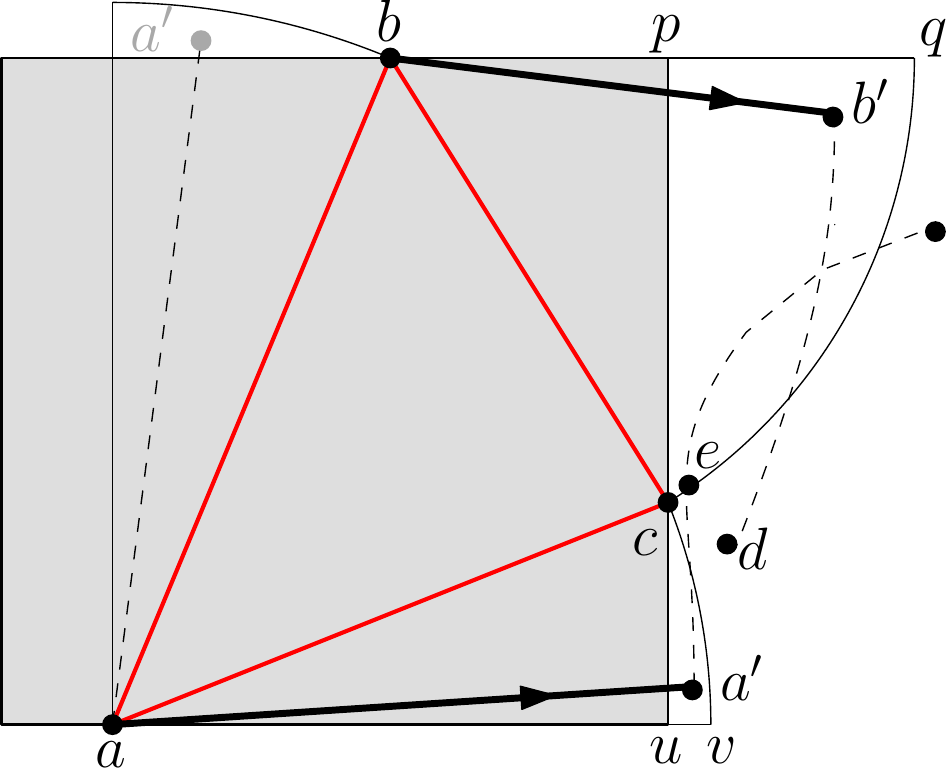} \\
(a) & (b) & (c)
\end{tabular}
\caption{\autoref{thm:y4main}: $(a, b)$ spans adjacent sides of $S$, and (a) $(a, a')$ lies counterclockwise from $(a, b)$, or (b) $(a, a')$ lies clockwise from $(a, b)$; (c) $(a, b)$ spans opposite sides of $S$.}
\label{fig:y4main}
\end{figure}
%%%%%%%%%%%%%%%%%%%%%%%%%%%%%%%%%Figure End
%
\paragraph{Case 1.} Consider first the simpler case when $a$ and $b$ lie on adjacent sides of $S$. Assume without loss of generality that 
$a$ and $b$ lie on the west and north sides of $S$ respectively, so $b \in Q_1(a)$. 

Assume first that $(a,a')$ lies counterclockwise from $(a, b)$. Since $S$ is empty of vertices, $a'$ must be above $b$. Refer to~\autoref{fig:y4main}(a).
Inequality~(\ref{eq:a'b}) together with~\autoref{lem:R} implies  
\begin{equation*}
d_R(a', b) \le \sqrt{2} \cdot d_2(a', b) < 2 \cdot d_2(a, b) , 
\label{eq:dR}
\end{equation*}
and similarly for $d_R(b, a')$. By~\autoref{lem:y4zigzag}, 
$p_R(a', b)$ may not cross $(a, b)$, therefore $p_R(a', b)$ exits $R(a', b)$ through its right side. This implies that 
the paths $p_a = (a,a') \oplus p_R(a', b)$ and $p_b = p_R(b, a')$ intersect. If $p_a$ and $p_b$ share a vertex,
define $x = y$ to be the common vertex; 
otherwise, let $(x, y)$ be a shortest side of the quadrilateral formed by the endpoints of the two crossing edges. 
\autoref{lem:R} tells us that the two crossing edges are no longer than $\max\{d_2(a, a'), d_2(a', b)\}$, and by inequalities~(\ref{eq:aa'}) and~(\ref{eq:a'b}) this quantity is no greater than $\sqrt{2} \cdot d_2(a, b)$. 
This along with~\autoref{lem:y4cross} implies that 
$d_{Y_4} (x,y) \le 3(2+\sqrt{2})\sqrt{2}\cdot d_2(a, b) = 6(1+\sqrt{2})\cdot d_2(a, b)$. 
These together show that  
 \begin{eqnarray*}
 d_{Y_4} (a,b)  & \le & d_2(a,a') + d_R(a',b) + d_R(b, a') + d_{Y_4} (x, y)  \\
 & \le & d_2(a, b) + 2 \cdot d_2(a, b) + 2 \cdot d_2(a, b) +  6(1+\sqrt{2})\cdot d_2(a, b) \\
 & = & (11 + 6\sqrt{2}) \cdot d_2(a, b)
 \end{eqnarray*}
Thus the theorem holds for this case. 

Assume now that $(a,a')$ lies clockwise from $(a, b)$. Refer to~\autoref{fig:y4main}(b). 
Let $\arr{(b, b')} \in Y_4$, with $b' \in Q(b, a)$. By definition, 
\begin{equation}
d_2(b, b') \le d_2(a, b)
\label{eq:bb'} 
\end{equation}
If $(b,b')$ lies clockwise from $(b, a)$, we find ourselves in a situation similar to the one depicted in~\autoref{fig:y4main}(a), with $a$ and $b$ switching roles. An analysis similar to the one above shows that the theorem holds for this case. So assume that $(b, b')$ lies counterclockwise from $(b,a)$, as depicted in~\autoref{fig:y4main}(b). In this case $(b, b')$ and $(a, a')$ cross in an interior point. Let $(x, y)$ be a shortest side of the quadrilateral with vertices $a$, $b'$, $a'$ and $b$.~\autoref{lem:y4cross}, along with inequalities (\ref{eq:aa'}) and (\ref{eq:bb'}), implies that $d_{Y_4} (x,y) \le 3(2+\sqrt{2})\cdot d_2(a, b)$. Thus we have that 
\begin{eqnarray*}
 d_{Y_4} (a,b)  & \le & d_2(a,a') + d_2(b, b') + d_{Y_4} (x, y)  \\
 & \le & d_2(a, b) + d_2(a, b) + 3(2+\sqrt{2})\cdot d_2(a, b) \\
 & = & (8 + 3\sqrt{2}) \cdot d_2(a, b)
 \end{eqnarray*}
So the theorem holds for this case as well. 

\paragraph{Case 2.}  Consider now the case where $a$ and $b$ lie on opposite sides of $S$. Recall that $(a, b)$ is one side of the 
triangle $\triangle abc$ enclosed in $S$. Assume without loss of generality that $a$ and $b$ lie on the south and north sides 
of $S$ respectively, and that $b \in Q_1(a)$. We further assume that $c \in Q_1(a)$; if this is not the case, we reverse the roles of $a$ and $b$ and rotate the vertex set $V$ by $\pi$ to make this assumption hold. 

The situation where $(a,a')$ lies counterclockwise from $(a, b)$ is similar to the one depicted in~\autoref{fig:y4main}(a) and the same analysis applies here as well. So assume that $(a,a')$ lies clockwise from $(a, b)$, as depicted in~\autoref{fig:y4main}(c). 
%If $(a,a')$ lies counterclockwise from $(a, b)$, the analysis used for the case depicted in~\autoref{fig:y4main}(a) applies here as well, so the theorem holds. So assume that $(a,a')$ lies clockwise from $(a, b)$, as depicted in~\autoref{fig:y4main}(c). 
Let $\arr{(b, b')} \in Y_4$, with $b' \in Q(b, a')$. 
By definition, 
\begin{equation}
d_2(b, b') \le d_2(b, a') < \sqrt{2} \cdot d_2(a, b) \mbox{~~(cf. inequality~(\ref{eq:a'b}))} 
\label{eq:bb'2}
\end{equation}
Also by the definition of $Y_4$, $d_2(a, a') \le d_2(a, c)$ and $d_2(b, b') \le d_2(b, c)$, therefore 
$a'$ is no higher and $b'$ is no lower than $c$. 
We first discuss the more complex situation where $(b, b')$ does not intersect $(a, a')$. We seek to identify two intersecting paths in 
$Y_4$, one that begins at $a'$ and extends toward $b'$, and one that begins with $(b, b')$ and then heads toward $a'$. 
By~\autoref{lem:y4zigzag}, $p_R(a', c)$ may not cross $(a, c)$, and similarly $p_R(b', c)$ may not cross $(b, c)$. This implies 
that $p_R(a', c)$ extends above and to the right of $c$, and $p_R(b', c)$ extends below and to the right of $c$.
% both paths intersect the horizontal through $c$. 

Let $e$ be the endpoint of $p_R(a', c)$ other than $a'$,  and let $d$ be the endpoint of $p_R(b', c)$ other than $b'$. 
If $p_R(a', c)$ and $p_R(b', c)$ intersect, we found our two intersecting paths. Otherwise, 
consider the more general case where $e$ lies left of $d$ (a similar analysis applies to the case where $e$ lies right of $d$). In this case, 
note that $p_R(e, b')$ is trapped underneath the path $p_{b} = (b, b') \oplus p_R(b',c)$, therefore $p_R(e, b')$ must intersect $p_b$. 
Thus we have identified two intersecting paths, 
$p_{a'} = p_R(a', c) \oplus p_R(e, b')$ and $p_{b} = (b, b') \oplus p_R(b',c)$. 
If $p_{a'}$ and $p_b$ share a vertex,
define $x = y$ to be the common vertex; otherwise, let $(x, y)$ be a shortest side of the quadrilateral 
formed by the endpoints of the two edges on $p_{a'}$ and $p_b$ that cross. 
Next we determine an upper bound on the length of these crossing edges, which together with~\autoref{lem:y4cross} will help determine an upper bound on the distance in $Y_4$ between $x$ and $y$. 

Let $p$ be the upper right corner of $S$. Let $q$ be the intersection between the horizontal through $p$ and the circle with center $b$ and radius $(b, c)$. Refer to~\autoref{fig:y4main}c. Similarly, let $u$ be the lower right corner of $S$ and let $v$ be the intersection between the horizontal through $u$ and the circle with center $a$ and radius $(a, c)$. 
First observe that $d_2(p, q) < d_2(p, c)$  (this follows immediately from the fact that $\angle{pcq} < \angle{bcq} = \angle{bqc}$), and similarly $d_2(u, v) < d_2(u, c)$. This implies  
\begin{eqnarray}
\nonumber d_2(c, q)  <  \sqrt{2} \cdot d_2(p, c) < \sqrt{2} \cdot d_2(a, b) \\
d_2(v, c)  < \sqrt{2} \cdot d_2(u, c) < \sqrt{2} \cdot d_2(a, b) 
\label{eq:cq} 
\end{eqnarray}
We use these inequalities, along with~\autoref{lem:R}, to establish the following upper bounds: each edge on $p_R(a', c)$ is no longer than $d_2(a', c) \le d_2(v,c) < \sqrt{2} \cdot d_2(a, b)$,  conform inequality~(\ref{eq:cq}); each edge on $p_R(e, b')$ is no longer than $d_2(e, b') \le d_2(c, b') < d_2(c, q) <  \sqrt{2} \cdot d_2(a, b)$, and similarly for each edge on $p_R(b', c)$. Inequality~(\ref{eq:bb'2}) shows that the same upper bound of $\sqrt{2} \cdot d_2(a, b)$ holds for $d_2(b, b')$ as well. We conclude that each of the two crossing edges on $p_{a'}$ and $p_b$ is no longer 
than $\sqrt{2} \cdot d_2(a, b)$. This along with~\autoref{lem:y4cross} implies that 
\begin{eqnarray}
 d_{Y_4} (x, y)  & \le 3(2+\sqrt{2})\sqrt{2}\cdot d_2(a, b) = 6(1+\sqrt{2})\cdot d_2(a, b)
\label{eq:xy2}
\end{eqnarray}
Also by~\autoref{lem:R} we have that $d_R(a',c) \le \sqrt{2} \cdot d_2(a', c)$, 
 $d_R(b', c) \le \sqrt{2} \cdot d_2(b', c)$, and 
$d_R(e,b') \le \sqrt{2} \cdot d_2(e, b') \le \sqrt{2} \cdot d_2(b', c)$. Summing up these inequalities yields 
 \begin{eqnarray*}
 d_R(a',c) + d_R(e, b') + d_R(b', c) & \le & \sqrt{2} \cdot (d_2(a', c) + d_2(b', c) + d_2(b', c)) \\
 & \le & \sqrt{2} \cdot (d_2(v, c) + d_2(q, c) + d_2(q, c)) 
  \end{eqnarray*}
Substituting the inequalities from~(\ref{eq:cq}) in the inequality above yields 
 \begin{eqnarray*}
 d_R(a',c) + d_R(e, b') + d_R(b', c) & \le & \sqrt{2} \cdot (\sqrt{2} \cdot d_2(u, c) + \sqrt{2} \cdot d_2(p, c) + \sqrt{2} \cdot d_2(p, c)) \\
 & = & 2 \cdot (d_2(u, p) + d_2(p, c)) < 4 \cdot d_2(a, b)
  \end{eqnarray*}
This latter inequality follows from the fact that each of $d_2(u, p)$ and $d_2(p, c)$ is bounded above by $d_2(a,b)$.
This together with inequalities~(\ref{eq:aa'}),~(\ref{eq:bb'2}) and~(\ref{eq:xy2}) yields 
 \begin{eqnarray*}
 d_{Y_4} (a,b)  & \le & d_2(a,a') + d_2(b, b') + (d_R(a',c) + d_R(e, b') + d_R(b', c)) + d_{Y_4}(x, y)  \\
 & \le & d_2(a, b) + \sqrt{2} \cdot d_2(a, b) + 4 \cdot d_2(a, b) + 6(1+\sqrt{2})\cdot d_2(a, b) \\
 & = & (11 + 7\sqrt{2}) \cdot d_2(a, b)
 \end{eqnarray*}
Thus the theorem holds for this case. The case where $(b, b')$ intersects $(a, a')$ is a special instance of the one discussed above, with the paths $p_{a'}$ and $p_R(b', c)$ reduced to null.  This concludes the proof. 
\end{proof}

\section{Conclusion}
In this paper we improve the upper bounds on the stretch factors of $Y_4^\infty$ and $Y_4$. The best known lower bound on the stretch factor of $Y_4^\infty$ is the one established in~\cite{BG+12} for $Del^\infty$, which is 
$\sqrt{4+2\sqrt{2}} \lesssim 2.62$. Narrowing the gap between this lower bound and the upper bound of $4.94$ established in this paper remains open. 

The second result of this paper reduces the upper bound on the stretch factor of $Y_4$ from $662.16$ to $54.62$. This bound might be improved with a more careful analysis that does not rely on the intermediate results from~\cite{BDD+12} employed by~\autoref{lem:y4cross}. We believe that the real stretch factor is much lower, and leave open reducing the 
upper bound further. 
%Regarding the lower bound, we have failed to construct an example of a $Y_4$ graph with stretch factor greater than $7$.  

%\small
%\bibliographystyle{plain}
%\bibliography{y4references}

\end{document}